\pdfoutput=1
\documentclass[lettersize,journal]{IEEEtran}
\usepackage{cite}
\usepackage{amsmath,amssymb,amsfonts}
\usepackage{amsthm}
\usepackage{algorithm}
\usepackage{algpseudocode}
\usepackage{graphicx}
\usepackage{epsfig} 
\usepackage{optidef}
\usepackage{textcomp}
\usepackage{csquotes}
\usepackage{url}
\usepackage{multicol}
\usepackage{color}
\usepackage{tabularray}

\newtheorem{theorem}{Theorem}
\newtheorem{crl}{Corollary}
\newtheorem{lemma}{Lemma}

\newtheorem{assumption}{Assumption}
\allowdisplaybreaks

\begin{document}

\title{Data-Driven Robust Reinforcement Learning Control of Uncertain Nonlinear Systems:\\ Towards a Fully-Automated, Insulin-Based\\ Artificial Pancreas}

\author{Alexandros Tanzanakis and John Lygeros
\thanks{Alexandros Tanzanakis and John Lygeros: Department of Information Technology and Electrical Engineering, ETH Zurich, Switzerland, {\tt\small \{atanzana,jlygeros\}@ethz.ch}}
\thanks{This research work was supported by the European Research Council (ERC) under the project OCAL, grant number 787845.}}

\markboth{Journal of \LaTeX\ Class Files,~Vol.~14, No.~8, August~2021}%
{Shell \MakeLowercase{\textit{et al.}}: A Sample Article Using IEEEtran.cls for IEEE Journals}


\maketitle

\begin{abstract}
In this paper, a novel robust tracking control scheme for a general class of discrete-time nonlinear systems affected by unknown bounded uncertainty is presented. By solving a parameterized optimal tracking control problem subject to the unknown nominal system and a suitable cost function, the resulting optimal tracking control policy can ensure closed-loop stability by achieving a sufficiently small tracking error for the original uncertain nonlinear system. The computation of the optimal tracking controller is accomplished through the derivation of a novel Q-function-based $\lambda$-Policy Iteration algorithm. The proposed algorithm not only enjoys rigorous theoretical guarantees, but also avoids technical weaknesses of conventional reinforcement learning methods. By employing a data-driven, critic-only least squares implementation, the performance of the proposed algorithm is evaluated to the problem of fully-automated, insulin-based, closed-loop glucose control for patients diagnosed with Type 1 and Type 2 Diabetes Mellitus. The U.S. FDA-accepted DMMS.R simulator from the Epsilon Group is used to conduct a comprehensive in silico clinical campaign on a rich set of virtual subjects under completely unannounced meal and exercise settings. Simulation results underline the superior glycaemic behavior achieved by the derived approach, as well as its overall maturity for the design of highly-effective, closed-loop drug delivery systems for personalized medicine.
\end{abstract}


\section{Introduction}
\IEEEPARstart{R}{obust} control approaches, rigorously studied in the classical feedback control framework \cite{b1}, are currently in the spotlight of modern learning-based control methods \cite{b2}. The existence of system uncertainty (ranging from unmodeled system dynamics and inexact system identification to external pertubations and disturbances) is well-known to seriously degrade the overall behavior and closed-loop stability of critical complex systems. To ensure an acceptable level of performance, a robust controller should guarantee closed-loop stability under any realization of uncertainty in effect. The leading works \cite{b3,b4,b5,b6} introduced the idea of solving a model-based robust stabilization problem for an uncertain linear system through the derivation of optimal control methods. Towards this direction, approximate dynamic programming (ADP) \cite{b7,b8}, by marrying reinforcement learning (RL) control with approximation methods, is widely used for effective optimal control of continuous-time and discrete-time dynamical systems.\\
For the case of uncertain systems, a rich family of ADP-based robust stabilization approaches has been proposed. These methods are mainly classified as follows: 1) via transformation into an optimal control problem for either the nominal \cite{b9,b10,b11,b12,b13,b14} or a virtual auxiliary system \cite{b15,b16,b17,b18} with a suitable cost function, 2) zero-sum game formulation for $H_{\infty}$ control \cite{b19,b20,b21,b22}, 3) guaranteed cost control \cite{b23,b24,b25}, 4) disturbance estimation and compensation design \cite{b26,b27,b28}, 5) sliding mode control \cite{b29,b30} and 6) robust ADP \cite{b31,b32,b33}. Despite the wide range of available methods, they inherit weakening assumptions which may limit their application on real complex systems. In particular, apart from very few results on uncertain discrete-time systems with nominal linear-time-invariant (LTI) or affine nonlinear dynamics \cite{b10,b14,b15,b21,b26,b33}, the vast majority of state-of-art approaches are applicable only to uncertain continuous-time systems. Furthermore, the complete knowledge of the system and disturbance dynamics (either assumed known apriori or obtained via system identification) is required to implement nearly all of the available methods. On a final note, the Policy Iteration (PI) \cite{b34,b35} algorithm is mainly utilized to implement the derived robust stabilization strategies. Even though PI enjoys relatively fast convergence to optimal control policies, it suffers from the strict requirement for an initial stabilizing control policy, which is challenging or even impossible to compute in the case of complex nonlinear system dynamics. On the contrary, Value Iteration (VI) \cite{b36,b37} benefits from an easy-to-realize initialization condition, at the expense of slower convergence compared to PI.\\
The design of high-performance ADP methods that integrate the merits of PI and VI algorithms has received increasing attention. To this end, the $\lambda$-PI algorithm \cite{b38,b39} conveniently enables a balancing weight on the trade-offs between PI and VI. $\lambda$-PI is presented and analyzed in \cite{b40,b41,b42} for finite Markov Decision Processes. Data-driven, randomized variants of $\lambda$-PI for contractive models are discussed in\cite{b43,b44}. A model-free $\lambda$-PI algorithm for optimal regulation of discrete-time deterministic linear systems is proposed in \cite{b45}. However, to the best of the authors knowledge, the algorithm has not been studied for the case of uncertain systems, remaining an important open problem.\\
Due to its outstanding theoretical advancements, learning-based control has found rigorous applications on personalized medicine e.g., on management and treatment of Diabetes Mellitus (DM) \cite{b46}. DM refers to a serious chronic metabolic disease characterized by critically elevated blood glucose levels. This can lead to severe life-threatening complications, including blindness, diabetic neuropathy and nephropathy, cardiovascular and chronic kidney disease. DM is mainly classified as Type 1 DM (T1DM) and Type 2 DM (T2DM) \cite{b47}. On the one hand, T1DM is caused by the autoimmune destruction of pancreatic $\beta$-cells which secrete a glucose- decreasing hormone called insulin. Therefore, T1DM patients are required to get lifelong exogenous doses of insulin. On the other hand, T2DM refers to the condition where, either pancreatic $\beta$-cells do not secrete enough insulin, or the human body does not respond effectively to the produced insulin. T2DM treament options mainly involve lifestyle changes and oral medication. However, if these options are not effective in regulating blood glucose levels, then exogenous administration of insulin is required, similar to T1DM patients.\\
The artificial pancreas (AP) \cite{b48} is considered a top-notch wearable drug delivery system for closed-loop glucose regulation. A conventional insulin-based AP consists of a control algorithm that automatically computes the required amounts of insulin to be continuously administered to the patient through an insulin pump. This is accomplished by collecting and utilizing patient-specific data such as blood glucose measurements received from a continuous glucose monitor (CGM), as well as meal and exercise information manually announced from the patient (e.g., carbohydrate (CHO) amounts, starting times of meals and exercises, exercise intensity etc.). The related state-of-art literature include results and approaches based on Model Predictive Control (MPC) \cite{b49,b50}, Proportional-Integral-Derivative (PID) \cite{b51}, $H_{\infty}$ control \cite{b52} and RL \cite{b53,b54,b55}. Interested readers can refer to the extensive review papers \cite{b56,b57,b58}. The vast majority of available methods are focused on T1DM, with very few results on T2DM (e.g., \cite{b59,b60}). Furthermore, commonly occured misestimation of meal and exercise information provided by the patient to the AP, in combination to significant metabolic delays of insulin action due to subcutaneous hormone delivery, can lead to dangerous insulin misdose and hence critically low or high blood glucose levels.\\
In recent years, some attempts have been made towards the design of a fully-automated, insulin-based AP systems for personalized, closed-loop glucose control, i.e., without requiring any meal and exercise announcements from the patient. Such approaches include MPC \cite{b61,b62,b63}, linear-parameter-varying (LPV) control \cite{b64}, feedforward control and disturbance observer design \cite{b65}, multiple model PID control \cite{b66}, model-free PID control \cite{b67} and RL control \cite{b68}. However, please note that all methods in \cite{b61,b62,b63,b64,b65,b67,b68} assume the construction of complex data-driven personalized models of the glucose regulatory system, which requires a significantly large amount of data to compute. Furthermore, the conducted in silico studies are restricted only to specific age cohorts (mainly to T1DM adult subjects). Last but not least, the proposed methods do not enjoy theoretical guarantees. Therefore, the problem of theoretically rigorous algorithmic methods for high-performance, fully-automated, insulin-based glucose control still remains a highly challenging open problem.  \\        
In this work:
\begin{itemize}
\item[1)] We propose a novel robust tracking control scheme for unknown discrete-time nonlinear systems affected by unknown bounded uncertainty. By solving a parameterized optimal tracking control problem subject to the unknown nominal system and an appropriate cost function, the resulting optimal tracking control policy can ensure closed-loop stability in terms of achieving a sufficiently small tracking error for the original uncertain system.
\item[2)] The computation of the optimal tracking control policy is achieved by deriving a novel Q-function-based variant of the $\lambda$-PI algorithm. The proposed algorithm is shown to benefit from theoretically rigorous monotonicity and convergence guarantees and avoid technical weaknesses of PI and VI. A data-driven, critic-only least squares approximation approach is employed, which enables a straightforward implementation solution for complex systems.
\item[3)] We evaluate the performance of the overall robust tracking control algorithm on the problem of a fully-automated, insulin-based, personalized glucose regulation. A U.S. FDA-accepted metabolic simulator is used to conduct highly challenging, in silico clinical studies on a rich variety of representantive T1DM and T2DM virtual subjects.
\end{itemize}

The structure of the paper is the following. The problem description is provided in Section II. The proposed robust tracking control scheme is derived and analyzed in Section III. A $\lambda$-PI algorithm that solves the parameterized optimal tracking control problem is presented in Section IV. A data-driven, critic-only least squares approximation approach is derived in Section V. The performed in silico clinical campaign is presented and discussed in Section VI. Finally, conclusions are given in Section VII.\\
\textbf{Notation.} $\mathbb{N}$ and $\mathbb{N}_{0}$ are the sets of natural numbers and natural numbers including $0$ respectively. $\mathbb{R_{+}}$ defines the set of non-negative real numbers. $\mathbb{S}^{N}_{++}$ defines the set of $N\times N$ symmetric positive definite matrices. $\|\cdot\|_2$ defines either the Euclidean vector norm or the matrix induced norm. By defining $\mathcal{X}\subset \mathbb{R}^{N}$, $\|f(x)\|_{\infty}$ defines the maximum norm of the function $f:\mathcal{X}\rightarrow \mathbb{R}_{+}$. \textit{Regarding measurement units:} mg/dL defines milligrams per deciliter and mg/dL/min refers to milligrams per deciliter per minute. U defines units, while U/5mins defines units per 5 minutes.  
\section{Problem Definition}
In this work, we consider the following class of discrete-time uncertain nonlinear systems given by
\begin{equation}
\label{eq:1}
x_{k+1} = f(x_k,u_k) + d(x_k),
\end{equation}
where $x\in \mathcal{X}\subset \mathbb{R}^{n}$ and $u \in \mathcal{U}\subset \mathbb{R}^{m}$ describe the system state and control input respectively, $k\in \mathbb{N}_0$ is the discrete-time index, $f:\mathcal{X}\times \mathcal{U}\rightarrow \mathcal{X}$ and $d:\mathcal{X}\rightarrow \mathcal{X}$ defines a general class of system uncertainty caused by external disturbances and perturbations, unmodeled system dynamics, system identification errors etc. Let also $r\in \mathcal{X}$ be a bounded reference signal described by an exosystem
\begin{equation}
\label{eq:2}
r_{k+1} = h(r_k),
\end{equation}
with $h:\mathcal{X}\rightarrow \mathcal{X}$.
\begin{assumption}
$\mathcal{X}$ and $\mathcal{U}$ are compact sets which contain the origin. $f(x,u)$ and $h(r)$ are Lipschitz continuous on $\mathcal{X}\times \mathcal{U}$ and $\mathcal{X}$ respectively with $f(0,0)=h(0)=0$. The mathematical expressions of $f(x,u)$, $d(x)$ and $h(r)$ are unknown. $d(x)$ is bounded by a known function $\Delta:\mathcal{X}\rightarrow \mathbb{R}_{+}$, i.e., $\big\|d(x)\big\|_{2}\leq \Delta(x)$ for all $x\in \mathcal{X}$ with $d(0)=\Delta(0)=0$.
\end{assumption}
We are interested in achieving robust tracking control of the uncertain nonlinear system \eqref{eq:1}, i.e., compute a state-reference feedback control policy $\mu(x,r):\mathcal{X}^{2}\rightarrow \mathcal{U}$ such that the closed-loop uncertain nonlinear system \eqref{eq:1} tracks a desired reference signal \eqref{eq:2} for any realization of uncertainty $d(\cdot)$ with $\big\|d(x)\big\|_{2}\leq \Delta(x)$ for all $x\in \mathcal{X}$.
\section{A Robust Tracking Control Scheme For Discrete-Time Uncertain Systems}
In this section, we show that the original robust tracking control problem discussed in Section II can be transformed into an optimal tracking control problem for the discrete-time nominal system
\begin{equation}
\label{eq:3}
x_{k+1} = f(x_k,u_k)
\end{equation}
subject to a suitable cost function. Therefore, we now focus on the computation of a state-reference feedback control policy $\mu(x,r)$, such that the closed-loop nominal system \eqref{eq:3} tracks a reference signal generated by \eqref{eq:2} through minimization of an appropriate infinite-horizon cost function. By defining $\mathcal{Z}= \mathcal{X}^{2}\times \mathcal{U}$, the infinite-horizon cost under consideration is defined as
\begin{align}
\label{eq:4}
J^{\mu}(x_0,r_0) =& \sum_{k=0}^{\infty} \gamma^{k}\bigg[l\big(x_k,r_k,\mu(x_k,r_k)\big)\nonumber \\
&+\Gamma \big(x_k,r_k,\mu(x_k,r_k)\big)\bigg] ,
\end{align}
where $\gamma \in (0,1]$ is the discount factor, $l:\mathcal{Z}\rightarrow \mathbb{R}_{+}$ is given by
\begin{equation}
\label{eq:5}
l\big(x,r,\mu(x,r)\big) = (x-r)^{T}S(x-r)+ [\mu(x,r)]^{T}R\mu(x,r)
\end{equation}
with $S \in \mathbb{S}_{++}^{n}$ and $R \in \mathbb{S}_{++}^{m}$, and $\Gamma:\mathcal{Z}\rightarrow \mathbb{R}_{+}$. The choice of function $\Gamma$ (to be defined in the sequel) is of critical importance, since it will ensure the robust tracking control of the original uncertain nonlinear system \eqref{eq:1}. A discount factor $\gamma \in (0,1)$ is required to ensure that \eqref{eq:4} is finite, while $\gamma=1$ can only be applied in tracking control problems characterized by asympotically stable reference dynamics \eqref{eq:2} \cite{b69,b70,b71}.
\begin{assumption}
The nominal system \eqref{eq:3} is controllable on $\mathcal{X}$. For the function $\Gamma(\cdot,\cdot,\cdot)$, it holds that $\Gamma(0,0,0)=0$.
\end{assumption}
We consider the case where the mathematical expressions of the system and reference dynamics \eqref{eq:1}-\eqref{eq:3} are unknown, although their values can be observed through simulations and experiments. To this end, we employ the Q-function formulation by defining $Q:\mathcal{Z}\rightarrow \mathbb{R}_{+}$ \cite{b72}. 
\begin{assumption}
$Q(\cdot,\cdot,\cdot)$ is a twice continuously differentiable function on $\mathcal{Z}$.
\end{assumption}
The optimal control policy $\mu^{\star}(x,r)$ that minimizes \eqref{eq:4} can be derived by firstly computing the optimal Q-function $Q^{\star}$, which satisfies the Bellman optimality equation \cite{b72,b73}
\begin{align}
\label{eq:6}
Q^{\star}(x_k,r_k,a_k)=&l(x_k,r_k,a_k)+\Gamma(x_k,r_k,a_k)\nonumber \\
&+\gamma Q^{\star}\big(x_{k+1},r_{k+1},\mu^{\star}(x_{k+1},r_{k+1})\big)
\end{align} 
where $a_k\in \mathcal{U}$. The optimal state-reference feedback control policy $\mu^{\star}$ is then computed by
\begin{equation}
\label{eq:7}
\mu^{\star}(x,r) = \underset{u}{\mathrm{argmin}}\enspace Q^{\star}(x,r,u).
\end{equation}
\begin{theorem}
Let Assumptions 1 to 3 hold. Assume that there exists a positive definite solution $Q^{\star}$ of \eqref{eq:6} with $\mu^{\star}(x,r)$ given by \eqref{eq:7} and $\mu^{\star}(0,0)=0$. Let also define the tracking error $e_k=x_k-r_k$ and 
\begin{align}
\label{eq:8}
\Gamma(x_k,r_k,a_k) =& \rho^{2} \Delta^{2}(x_k)+\frac{1}{4}\gamma \big\| \nabla_x Q^{\star}\big(x_{k+1},r_{k+1},\nonumber \\
&\mu^{\star}(x_{k+1},r_{k+1})\big) \big\|^{2}_{2}
\end{align}
where $\nabla_x Q^{\star}:\mathcal{Z}\rightarrow \mathbb{R}^{n}$ is the gradient of $Q^{\star}$ with $\nabla_x Q^{\star}(0,0,0)=0_{n\times 1}$, $x_{k+1} = f(x_k,a_k)$ and $\rho\geq 1$. Then, for all time steps $s\geq k+1:$
\begin{itemize}
\item[1)] If $\gamma=1$, $\mu^{\star}(x,r)$ can make the tracking error $e_s$ for the nominal system \eqref{eq:3} locally asymptotically stable. Otherwise, $\mu^{\star}(x,r)$ can make $e_s$ for \eqref{eq:3} sufficiently small by setting $\gamma$ sufficiently close to $1$.
\item[2)] Let the following condition
\begin{align}
\label{eq:9}
\rho^{2} \Delta^{2}(x_s) \geq & \gamma \big\|d(x_s)\big\|^{2}_{2}+\frac{1}{2}\gamma [d(x_s)]^{T}\nabla_{xx}Q^{\star}\big(x_{s+1},r_{s+1}, \nonumber \\
&\mu^{\star}(x_{s+1},r_{s+1})\big)d(x_s)
\end{align}
hold, where $\nabla_{xx}Q^{\star}:\mathcal{Z}\rightarrow \mathbb{R}^{n\times n}$ is the Hessian of $Q^{\star}$ and $x_{s+1}=f\big(x_s,\mu^{\star}(x_s,r_s)\big)$. Then, if $\gamma=1$, $\mu^{\star}(x,r)$ can make the tracking error $e_s$ for the uncertain system \eqref{eq:1} locally asymptotically stable. Otherwise, $\mu^{\star}(x,r)$ can make $e_s$ for \eqref{eq:1} sufficiently small by setting $\gamma$ sufficiently close to $1$.
\end{itemize}
\end{theorem}
\begin{proof}
See Appendix A.
\end{proof}
\textit{Remark 1:} Theorem 1 introduces an explicit functional form for $\Gamma$, based on which the control policy $\mu^{\star}$ ensures optimal tracking control and a sufficiently small tracking error for the nominal nonlinear system \eqref{eq:3}. Furthermore, if condition \eqref{eq:9} additionally holds, then $\mu^{\star}$ can also guarantee a sufficiently small tracking error for the uncertain nonlinear system \eqref{eq:1}. We note that condition \eqref{eq:9} translates to the computation of a suitable $\rho \geq 1$ such that the condition is satisfied. This makes the resulting optimal tracking control problem parameterized with respect to $\rho$. Moreover, since the mathematical expression of $d$ is unknown, we cannot evaluate condition \eqref{eq:9} directly. To overcome this issue, we proceed with the evaluation of the following condition
\begin{align}
\label{eq:10}
\rho^{2}\Delta^{2}(x_s) \geq & \gamma \Delta^{2}(x_s)+\frac{1}{2}\gamma \big\|\nabla_{xx} Q^{\star}\big(x_{s+1},r_{s+1},\nonumber \\
&\mu^{\star}(x_{s+1},r_{s+1})\big)\big\|_{2}\Delta^{2}(x_s)
\end{align}
for all $s\geq k+1$. It is easy to observe that when \eqref{eq:10} holds, then \eqref{eq:9} also holds since
\begin{align*}
\rho^{2}\Delta^{2}(x_s) \geq & \gamma \Delta^{2}(x_s) +\frac{1}{2}\gamma \big\|\nabla_{xx} Q^{\star}\big(x_{s+1},r_{s+1},\nonumber \\
&\mu^{\star}(x_{s+1},r_{s+1})\big)\big\|_{2}\Delta^{2}(x_s)\nonumber \\
\geq & \gamma \big\|d(x_s)\big\|^{2}_{2} +\frac{1}{2}\gamma \big\|[d(x_s)]^{T}\big\|_{2}\big\|\nabla_{xx} Q^{\star}\big(x_{s+1},\nonumber \\
&r_{s+1},\mu^{\star}(x_{s+1},r_{s+1})\big) \big\|_{2}\big\|d(x_s)\big\|_{2}\nonumber \\
\geq & \gamma \big\|d(x_s)\big\|^{2}_{2} +\frac{1}{2}\gamma \big\|[d(x_s)]^{T} \nabla_{xx} Q^{\star}\big(x_{s+1},\nonumber \\
&r_{s+1},\mu^{\star}(x_{s+1},r_{k+1})\big)d(x_s)\big\|_{2}\nonumber \\
\geq & \gamma \big\|d(x_s)\big\|^{2}_{2} +\frac{1}{2}\gamma [d(x_s)]^{T}\nabla_{xx} Q^{\star}\big(x_{s+1},\nonumber \\
&r_{s+1},\mu^{\star}(x_{s+1},r_{s+1})\big)d(x_s),
\end{align*}
where we utilized the fact that $\big\|d(x)\big\|_{2}\leq \Delta(x)$ and exploited simple norm properties.
\section{A $\lambda$-PI Algorithm for Parameterized Optimal Tracking Control}
In this section, we propose a novel Q-function-based $\lambda$-PI algorithm that solves the parameterized optimal tracking control problem presented in Section III. According to Algorithm 1, starting from a suitable $Q^{0}\geq 0$ (to be discussed in the sequel), the proposed algorithm proceeds with the iterative operations of policy improvement \eqref{eq:11} and policy evaluation \eqref{eq:12}, until convergence of the Q-function. The policy evaluation scheme \eqref{eq:12} introduces an additional parameter $\lambda \in (0,1)$ which provides a balancing weight between the conventional VI and PI algorithms. We note that VI and PI can be obtained by setting $\lambda =0$ and $\lambda =1$ in \eqref{eq:12} for all $i\geq 0$ respectively \cite{b34,b35,b36,b37}.
\begin{algorithm}        
\caption{A $\lambda$-PI algorithm that solves the parameterized optimal tracking control problem.}\label{euclid}
        \begin{algorithmic}[1]
            \State  \textbf{Initialization:} Choose $Q^{0}(x,r,a)\geq 0$. Set $i=0$ and $\tau\geq 0$.
            \State\textbf{Policy Improvement:} 
\begin{equation}
\label{eq:11}
\mu^{i}(x,r) = \underset{u}{\mathrm{argmin}}\enspace Q^{i}(x,r,u).
\end{equation}
            \State \textbf{Policy Evaluation:} Solve for $Q^{i+1}\geq 0$,
\begin{align}
\label{eq:12}
&Q^{i+1}(x,r,a)= l(x,r,a)+ \Gamma^{i}(x,r,a)\nonumber \\
&+\lambda \gamma Q^{i+1}\big(x',r',\mu^{i}(x',r')\big)\nonumber \\
&+(1-\lambda) \gamma Q^{i}\big(x',r',\mu^{i}(x',r')\big),
\end{align}
where $\Gamma^{i}(x,r,a)=\rho^{2}\Delta^{2}(x) +\frac{1}{4}\gamma \big\|\nabla_x Q^{i}\big(x',r',\mu^{i}(x',$ $r')\big)\big\|^{2}_{2}$, $\rho \geq 1$, $\lambda \in (0,1)$, $\gamma \in (0,1]$, $x'=f(x,a)$ and $r'=h(r)$.
             \State \textbf{Termination of Learning Phase:} If $\big\|Q^{i+1}-Q^{i}\big\|_{\infty}\leq \tau$, terminate. Otherwise, set $i=i+1$ and go to Step $2$.
        \end{algorithmic}
    \end{algorithm}

\begin{theorem}
Let Assumptions 1 to 3 hold. Consider the sequences $\{Q^{i}(x,r,a)\}_{i\in \mathbb{N}}$ and $\{\mu^{i}(x,r)\}_{i\in \mathbb{N}_0}$ generated by Algorithm $1$. Assume that the initialization condition
\begin{align}
\label{eq:13}
Q^{0}(x,r,a) \geq & l(x,r,a) +\Gamma^{0}(x,r,a) + \gamma Q^{0}\big(x',r',\mu^{0}(x',r')\big)
\end{align}
holds for all $(x,r,a)\in \mathcal{Z}$ and the following condition
\begin{align}
\label{eq:14}
\big\|\nabla_x Q^{i}\big(x,r,\mu^{i}(x,r)\big)\big\|^{2}_{2} \leq \big\|\nabla_x Q^{i-1}\big(x,r,\mu^{i-1}(x,r)\big) \big\|^{2}_{2}
\end{align}
holds for all $i\geq 1$ and $(x,r)\in \mathcal{X}^{2}$. Then, for all $(x,r,a)\in \mathcal{Z}$
\begin{enumerate}
\item[1)]
\begin{align}
\label{eq:15}
Q^{i+1}(x,r,a) \leq & l(x,r,a) +\Gamma^{i}(x,r,a)+ \nonumber \\
& \gamma Q^{i}\big(x',r',\mu^{i}(x',r')\big) \leq  Q^{i}(x,r,a).
\end{align}
\item[2)] 
\begin{align*}
\lim_{i\to \infty} Q^{i}(x,r,a) &= Q^{\star}(x,r,a),\\
\lim_{i\to \infty} \mu^{i}(x,r) &= \mu^{\star}(x,r).
\end{align*}
\end{enumerate}
\end{theorem}
\begin{proof}
See Appendix B.
\end{proof}
The following corollary reveals the impact of parameter $\lambda \in (0,1)$ to the resulting quality of Q-function estimates.
\begin{crl}
Let Assumptions 1 to 3 hold and $Q(x,r,a)\geq 0$ satisfy
\begin{align}
\label{eq:16}
Q(x,r,a) \geq & l(x,r,a) + \Gamma^{Q}(x,r,a) + \gamma Q\big(x',r',\mu(x',r')\big)
\end{align}
where $\Gamma^{Q}(x,r,a) =\rho^{2}\Delta^{2}(x)+ \frac{1}{4}\gamma \big\|\nabla_x Q\big(x',r',\mu(x',r')\big)\big\|^{2}_{2}$ and $\mu(x,r) =\underset{u}{\mathrm{argmin}}\enspace Q(x,r,u)$. Furthermore, let $0< \lambda_1\leq \lambda_2 < 1$ and assume that there exists a solution $Q_{\Lambda}(x,r,a)$ of the following equation
\begin{align*}
Q_{\Lambda}(x,r,a) =& l(x,r,a) + \Gamma^{Q}(x,r,a) +\Lambda \gamma Q_{\Lambda}\big(x',r',\mu(x',r')\big)\nonumber \\
&+ (1-\Lambda)\gamma Q\big(x',r',\mu(x',r')\big)
\end{align*}
with $\Lambda \in \{\lambda_1,\lambda_2\}$. Then, $Q_{\lambda_2}(x,r,a) \leq Q_{\lambda_1}(x,r,a)$ for all $(x,r,a)\in \mathcal{Z}$.
\end{crl}
\begin{proof}
See Appendix C.
\end{proof}
Based on Corollary 2, by utilizing a larger value of $\lambda$, the proposed $\lambda$-PI algorithm can compute improved Q-function estimates.  

\textit{Remark 2:} Conditions \eqref{eq:13} and \eqref{eq:14} are sufficient conditions. According to Theorem 2, since the sequence $\{Q^{i}(x,r,a)\}_{i\in \mathbb{N}_0}$ is monotonically non-increasing, we can exploit the rigorous theoretical properties of convnetional VI-based methods \cite{b36,b72}. To ensure that \eqref{eq:13} holds, it suffices to set $\lambda$ a value either equal or sufficiently close to 0 for $i=0$ and define $Q^{0}(x,r,a)$ as a sufficiently large, twice continuously differentiable, positive definite function. Furthermore, according to Corollary 1, we can then gradually increase the value of $\lambda$ in subsequent iterations $i>0$ to benefit from improved Q-function estimates $Q^{i}(x,r,a)$ and ensure that \eqref{eq:14} holds. Therefore, the proposed algorithm has the potential to ensure monotonic convergence to improved solutions and/or faster speed of convergence compared to conventional VI-based algorithms, as shown on extensive in silico clinical studies for fully-automated, closed-loop glucose regulation in T1DM and T2DM (Section VI). However, we note that an appropriate rate of increase for $\lambda$ is mainly dependent on the complexity of the respective system.
\section{A Data-Driven Implementation}
We now proceed with the derivation of a data-driven, critic-only implementation for Algorithm 1. Towards this direction, we introduce the following linearly parameterized approximation for $Q^{i}(x,r,a)$ on $\mathcal{Z}$
\begin{equation*}
Q^{i}(x,r,a) = [\Phi(x,r,a)]^{T} w^{i}+e^{i}(x,r,a),
\end{equation*}  
where $w^{i} \in \mathbb{R}^{K}$ is the weight vector, $\Phi(x,r,a)\in \mathbb{R}^{K}$ is a linearly independent polynomial basis function vector and $e^{i}(x,r,a)\in \mathbb{R}$ is the related function approximation error. Based on the Stone-Weierstrass Theorem, due to the fact that $\mathcal{Z}$ is compact, it holds that $\lim_{K\to \infty} e^{i}(x,r,a) = 0$ \cite{b74,b75}. However, since $w^{i}$ is unknown, we then set
\begin{equation*}
\hat{Q}^{i}(x,r,a) = [\Phi(x,r,a)]^{T}\hat{w}^{i},
\end{equation*}
where $\hat{w}^{i}\in \mathbb{R}^{K}$ defines an approximation of the weight vector $w^{i}$. Hence, the policy improvement \eqref{eq:11} and policy evaluation \eqref{eq:12} steps in Algorithm 1 are now given by
\begin{equation}
\label{eq:17}
\hat{\mu}^{i}(x,r) = \underset{u}{\mathrm{argmin}}\enspace \hat{Q}^{i}(x,r,u)
\end{equation}
and
\begin{align}
\label{eq:18}
\epsilon^{i+1}(x,r,a) =& \hat{Q}^{i+1}(x,r,a)-l(x,r,a)-\hat{\Gamma}^{i}(x,r,a)\nonumber \\
&-\lambda^{i}\gamma \hat{Q}^{i+1}\big(x',r',\hat{\mu}^{i}(x',r')\big)\nonumber \\
& -(1-\lambda^{i})\gamma\hat{Q}^{i}\big(x',r',\hat{\mu}^{i}(x',r')\big)\nonumber \\
=& \big[\Phi(x,r,a)-\lambda^{i}\gamma \Phi\big(x',r',\hat{\mu}^{i}(x',r')\big)\big]^{T}\hat{w}^{i+1}\nonumber \\
& -l(x,r,a)\nonumber -\hat{\Gamma}^{i}(x,r,a)-(1-\lambda^{i})\gamma \nonumber \\
& \cdot \big[\Phi\big(x',r',\hat{\mu}^{i}(x',r')\big)\big]^{T}\hat{w}^{i}
\end{align} 
respectively, where
\begin{align*}
\hat{\Gamma}^{i}(x,r,a) =& \rho^{2} \Delta^{2}(x)+\frac{1}{4}\gamma \big\|\nabla_{x}\hat{Q}^{i}(x',r',\hat{\mu}^{i}(x',r')\big)\big\|^{2}_{2},
\end{align*}
\begin{algorithm}
        \caption{The proposed R$\lambda$PI-LS algorithm.}\label{euclid}
        \begin{algorithmic}[1]
             \State \textbf{Parameterization:} Set $\rho= 1$.
		 \State \textbf{Initialization:} Choose $\hat{Q}^{0}(x,r,a)\geq 0$ based on Remark 2. Set $i=0$ and $\tau\geq 0$.
               \State \textbf{Policy Improvement:} Compute $\hat{\mu}^{i}(x,r)$ based on \eqref{eq:17}.
                \State \textbf{Adaptive Balancing:} Choose $\lambda^{i}$ based on Remark 2.
             \State  \textbf{Data Collection:} Construct data buffer $S^{i}$ based on \eqref{eq:19}.
             \State \textbf{Policy Evaluation:} Solve \eqref{eq:20} for $\hat{w}^{i+1}$.
\State \textbf{Termination of Learning Phase:}\newline If $\underset{b}{\mathrm{max}} \big|\hat{Q}^{i+1}(x_b,r_b,a_{b})-\hat{Q}^{i}(x_b,r_b,a_{b})\big|> \tau$, set $i=i+1$ and go to Step 3. Otherwise, set $\hat{Q}^{\star} =\hat{Q}^{i+1}$ and $\hat{\mu}^{\star}=\hat{\mu}^{i+1}$, computed by \eqref{eq:17}.
\State \textbf{Robustness Check:} Evaluate conditions of Theorem 1: If $\hat{Q}^{\star}$ is a twice continuously differentiable, positive definite function on $\mathcal{Z}$ and the pair $\{\hat{Q}^{\star}(x,r,a),\hat{\mu}^{\star}(x,r)\}$ with $\mu^{\star}(0,0)=0$ satisfy \eqref{eq:6} and \eqref{eq:10} for a sufficiently large number of time steps $C \in \mathbb{N}$ during the operation of the nominal system (3), return $\hat{\mu}^{\star}$ as the approximate robust controller of the uncertain system \eqref{eq:1}. Otherwise, set $\rho=\rho+1$ and go to Step 2.       
        \end{algorithmic}
    \end{algorithm}
$\lambda^{i}$ is the adaptive balancing parameter according to Remark 2 and $\epsilon^{i+1}(x,r,a)$ is the residual error due to the approximation errors $e^{i}$ on $\hat{Q}^{i}$, $e^{i+1}$ on $\hat{Q}^{i+1}$ and $\nabla_{x} e^{i}$ on $\nabla_{x} \hat{Q}^{i}$. We note that, due to Assumption 3, $Q^{i}$, $\nabla_{x} Q^{i}$ and $\nabla_{xx} Q^{i}$ are bounded on the compact set $\mathcal{Z}$. Therefore, the associated approximation errors $e^{i}$, $\nabla_{x} e^{i}$ and $\nabla_{xx} e^{i}$ are also bounded on $\mathcal{Z}$. \\
According to \eqref{eq:18}, we can compute $\hat{w}^{i+1}$ by gathering data along the operation of the nominal system \eqref{eq:3}. To this end, for all $i\geq 0$, let
\begin{align}
\label{eq:19}
S^{i} =& \big\{x_b,r_b,a_{b},x'_{b},r'_{b},\hat{\mu}^{i}(x'_{b},r'_{b})\big\}_{b=1}^{B}
\end{align}
be a data buffer constructed through interactions with the unknown nominal system \eqref{eq:3} and reference dynamics \eqref{eq:2} with $x'_{b} = f(x_{b},a_{b})$, $r'_{b} = h(r_b)$ and $B\in \mathbb{N}$ defines the size of the buffer. The residual error is then given by
\begin{align*}
\epsilon^{i+1}(x_b,r_b,a_b) =&  [\Phi(x_b,r_b,a_b)-\lambda^{i}\gamma \Phi\big(x'_b,r'_b,\hat{\mu}^{i}(x'_b,r'_b)\big)]^{T} \\
&\cdot \hat{w}^{i+1}-l(x_b,r_b,a_b)-\hat{\Gamma}^{i}(x_b,r_b,a_b) \\
&-(1-\lambda^{i})\gamma [\Phi\big(x'_b,r'_b,\hat{\mu}^{i}(x'_b,r'_b)\big)]^{T}\hat{w}^{i}
\end{align*} 
for $b=1,\ldots,B$. Therefore, the unknown vector $\hat{w}^{i+1}$ can then be computed by minimizing the sum of residual errors
\begin{equation*}
\min \sum_{b=1}^{B} (\epsilon^{i+1}_{b})^{2}.
\end{equation*}
A least-squares implementation scheme is thus given by
\begin{equation}
\label{eq:20}
\hat{w}^{i+1} = \bigg[[\Psi^{i}]^{T}\Psi^{i}\bigg]^{-1}[\Psi^{i}]^{T} z^{i}
\end{equation}
where $z^{i} = \begin{bmatrix} z^{i}_{1} & \cdots & z^{i}_{B}\end{bmatrix}^{T}$, $\Psi^{i} = \begin{bmatrix} \Psi^{i}_{1} & \cdots & \Psi^{i}_{B}\end{bmatrix}^{T}$, $\Psi^{i}_{b} = \Phi(x_{b},r_b,a_{b})-\lambda^{i} \gamma \Phi \big(x'_{b},r'_b,\hat{\mu}^{i}(x'_b,r'_b)\big)$, $z^{i}_{b}= l(x_b,r_b,a_b)+\hat{\Gamma}^{i}(x_b,r_b,a_b) +(1-\lambda^{i})\gamma [\Phi\big(x'_b,r'_b,\hat{\mu}^{i}(x'_b,r'_b)\big)]^{T}\hat{w}^{i}$. Finally, by considering the robust tracking control scheme introduced in Section III, Algorithm 2 shows the overall data-driven, robust $\lambda$-PI algorithm, which we refer to as R$\lambda$PI-LS.\\
\textit{Remark 3:} For all $i$, we note that the least squares scheme \eqref{eq:20} for computing $\hat{w}^{i+1}$ requires the inverse of the matrix $[\Psi^{i}]^{T}\Psi^{i} \in \mathbb{R}^{K\times K}$ to exist. This can be achieved by ensuring $\Psi^{i} \in \mathbb{R}^{B\times K}$ to be full column rank, for instance by applying  randomized or even general policies $a\in \mathcal{U}$ to the system \eqref{eq:3} which are significantly different from $\hat{\mu}^{i}$ (e.g., off-policy learning \cite{b76,b77}).

\section{In Silico Clinical Studies}
We now proceed by evaluating the tracking control capibilities of the proposed R$\lambda$PI-LS algorithm (Algorithm 2) on the problem of fully-automated, insulin-based, closed-loop glucose control. The U.S. FDA-accepted DMMS.R simulation software (v1.2.1) from the Epsilon Group \cite{b78,b79} has been employed to conduct extensive in silico clinical studies for personalized disease management and treatment on representantive virtual T1DM (11 adults, 11 adolescents, 11 children) and T2DM (11 adults) subjects. To replicate a realistic AP system configuration, we utilize commercial CGM and insulin pump settings \cite{b80} (e.g., CGM measurements and rapid-acting insulin infusion are provided every 5 minutes).

\subsection{Algorithmic Configuration}  
The state vector is defined as $x_k =\begin{bmatrix} x_{1,k} & x_{2,k} \end{bmatrix}^{T}$, where $x_{1,k}$ refers to the CGM measurements for the blood glucose concentration at time step $k$ [mg/dL] and $x_{2,k}$ computes the $30$-minute rate of change in blood glucose, i.e., $x_{2,k}=(x_{1,k}-x_{1,k-6})/30$ [mg/dL/min]. For the purpose of effective glucaemic control, we set $r_k=r^{\star}=120$ mg/dL, which is a widely used, clinically validated glucose setpoint \cite{b81,b82}. The control policy $\mu(x,r)$ defines the amount of rapid-acting insulin to be infused to the patient [U/5mins]. Furthermore, by setting $\mathcal{S}=1$ and $R=300$, the stage cost function \eqref{eq:5} becomes $l\big(x_k,r_k,\mu(x_k,r_k)\big)= (x_{1,k}-r_k)^{T}\mathcal{S}(x_{1,k}-r_k)+\mu^{T}(x_k,r_k)R\mu(x_k,r_k)$. The discount factor is given by $\gamma=0.95$, while the threshold value for algorithmic convergence is $\tau=10^{-10}$, the size of the data buffer $B_{\lambda PI}=144$ (which defines $12$-hour data collection intervals), the adaptive balancing parameter $\lambda^{i}= \tanh(0.7\ln(i+1))$ and the non-negative bound function $\Delta(x) = \sqrt{90x^{T}x}$. At the beginning of an in silico clinical trial, the state variable $x_{1,0}$ is initialized within the range $[70,180]$ mg/dL according to a uniform distribution, while for the computation of $x_{2,k}$ we set $x_{1,k-6}=0$ until there are available data to use.\\
The considered family of Q-functions $\hat{Q}^{i}(x,r,a)$ is given by the sum of unique basis elements derived from the polynomial function $z^{T}\hat{W}^{i}z$ for all $i$, where $z =\begin{bmatrix} x_1 & x_2 & x^{2}_1 & x^{2}_{2} & r & r^{2} & a \end{bmatrix}^{T} \in \mathbb{R}^{7}$ and $\hat{W}^{i}\in \mathbb{S}^{7\times 7}$ defines the unknown symmetric weight matrix. Therefore, it yields that $\hat{Q}^{i}(x,r,a) =[\Phi(x,r,a)]^{T}\hat{w}^{i}$, where $\Phi(x,r,a)\in \mathbb{R}^{28}$ and $\hat{w}^{i} \in \mathbb{R}^{28}$. We initialize $\hat{Q}^{0}(x,r,a)$ with a sufficiently large, positive definite function. We underline the fact that the weight element associated with the basis function $a^{2}$ on $\hat{Q}^{0}(x,r,a)$ should hold a significanly larger value (i.e., approximately $10^5\times$ higher) compared to all other weights, so that the initial policy improvement \eqref{eq:11} computes a reasonable insulin policy $\hat{\mu}^{0}(x,r)$. The exploration policy $a_{k}$ is defined as
\begin{align*}
a_{k} = \begin{cases} \hat{\mu}^{i}(x_k,r_k)+n_{k}, \text{ if } i=0,\\
\big(\hat{\mu}^{i}(x_k,r_k)+\hat{\mu}^{i-1}(x_k,r_k)\big)/2+n_{k}, \text{ if } i>0,
\end{cases}
\end{align*}
where $n_{k}$ is a randomized quantity according to a uniform distribution over $[3\cdot 10^{-4},6\cdot 10^{-4}]$. After algorithmic convergence, the robustness conditions (Step 8 in Algorithm 2) are evaluated for $C=576$ time steps, which translates to a total of 2 days. Finally, to fully assess the proposed R$\lambda$PI-LS algorithm, we compare its performance with the LS-based robust VI algorithm (obtained by setting $\lambda^{i}=0$ for all $i$ in Algorithm 2), which is referred to as RVI-LS and utilizes the exact same configuration discussed above.

\subsection{Clinical Simulation Settings and Results}
As part of the conducted in silico studies, we proceed with the computation of clinically validated, glycaemic control metrics \cite{b83} for the respective virtual T1DM and T2DM populations. These include:
\begin{itemize}
\item The mean, minimum and maximum CGM-based glucose concentration,
\item The time spent (as a percentage) in normoglycaemia ($x_{1,k}\in [70, 180]$ mg/dL), mild hypoglycaemia ($x_{1,k} \in [50, 70)$ mg/dL), severe hypoglycaemia ($x_{1,k}<50$ mg/dL), mild hyperglycaemia ($x_{1,k} \in (180, 250]$ mg/dL) and severe hyperglycaemia ($x_{1,k}> 250$ mg/dL),
\item The critical low and high blood glucose indices (referred to as LBGI and HBGI respectively), which express the frequency and magnitude of low and high glucose values respectively \cite{b84,b85},
\item The total administered rapid-acting insulin during a day (TDI), as well as the number of iterations until algorithmic convergence.
\end{itemize}
Simulation results are presented in the [mean value $\pm$ standard deviation] format. In order to execute both R$\lambda$PI-LS and RVI-LS algorithms, we conduct an in silico clinical trial for all considered T1DM and T2DM subjects by following a nominal daily meal and exercise configuration similar to \cite{b86}. The nominal meal configuration includes 6 meals which are scheduled to start daily at $[07\text{:}00, 10\text{:}00, 13\text{:}00, 15\text{:}00, 18\text{:}00, 23\text{:}00]$, contain $[70, 30, 90, 30, 90, 25]$ grams of carbohydrates (CHO) and last $[30, 15, 45, 15, 45, 20]$ minutes respectively. Furthermore, the nominal exercise configuration includes a session which starts at $16\text{:}00$, is characterized by moderate intensity and lasts 30 minutes. To replicate realistic daily lifestyle conditions, we introduce dynamic variability on every meal and exercise occurence of the abovementioned nominal configuration, with a range defined as follows: $1)$ $[-15,15]$ minutes on the meal starting time, $2)$ $[-15\%, 15\%]$ on the amount of CHO, $3)$ $[-15\%,15\%]$ on meal duration, $4)$ $[-15, 15]$ minutes on exercise starting time, $5)$ random selection of the exercise intensity mode from [light, moderate, intense] and $6)$ $[-20\%, 20\%]$ on exercise duration. The employed variability profile follows suitable uniform distributions.
\definecolor{Gray}{rgb}{0.498,0.498,0.498}
\begin{table*}[t]
\centering
\caption{GLYCAEMIC CONTROL PERFORMANCE OF THE RESPECTIVE ALGORITHMS (LEARNING PHASE).}
\definecolor{Silver}{rgb}{0.749,0.749,0.749}
\definecolor{Black}{rgb}{0,0,0}
\scalebox{0.85}{\begin{tblr}{
  cells = {c},
  cell{1}{1} = {c=13}{},
  cell{7}{1} = {c=13}{},
  hlines,
  vlines = {Silver},
  vline{2} = {-}{Black},
}
{\textbf{~ R$\lambda$PI-LS}\\\textbf{~}} &                                                &                                               &                                               &                                                   &                                                &                                                  &                                                 &                                                   &               &               &                                  &                                                            \\
\textbf{~}                       & {\textbf{BG}\\\textbf{mean}\\\textbf{[mg/dL]}} & {\textbf{BG}\\\textbf{min}\\\textbf{[mg/dL]}} & {\textbf{BG}\\\textbf{max}\\\textbf{[mg/dL]}} & {\textbf{\% in}\\\textbf{target}\\\textbf{range}} & {\textbf{\% in}\\\textbf{mild}\\\textbf{hypo}} & {\textbf{\% in}\\\textbf{severe}\\\textbf{hypo}} & {\textbf{\% in}\\\textbf{mild}\\\textbf{hyper}} & {\textbf{\% in}\\\textbf{severe}\\\textbf{hyper}} & \textbf{LBGI} & \textbf{HBGI} & {\textbf{TDI}\\\textbf{[U/day]}} & {\textbf{iterations}\\\textbf{till}\\\textbf{convergence}} \\
{\textbf{T1ADU}\\\textbf{~}}     & 152±8                                          & 83±17                                         & 204±16                                        & 82.6±4.3                                          & 1±0.9                                          & 0±0                                              & 16.4±3.4                                     & 0±0                                               & 0.32±0.18     & 2.09±0.91     & 45.3±9.4                         & 122±12                                                     \\
\textbf{T1ADO}                   & 150±12                                         & 81±15                                         & 201±15                                        & 78.2±5.7                                          & 0.5±0.3                                        & 0±0                                              & 21.3±5.4                                        & 0±0                                               & 0.22±0.19     & 2.31±0.75     & 41.2±9.5                         & 134±14                                                     \\
\textbf{T1CHIL}                  & 146±9                                          & 79±13                                         & 199±12                                        & 77.3±4.2                                          & 0.6±0.4                                        & 0±0                                              & 22.1±3.8                                        & 0±0                                               & 0.29±0.24     & 2.38±0.61     & 33.5±8.4                         & 128±10                                                     \\
\textbf{T2ADU}                   & 142±10                                         & 80±11                                         & 208±17                                        & 82.5±5.1                                          & 0.5±0.5                                        & 0±0                                              & 17±4.6                                          & 0±0                                               & 0.21±0.15     & 2.12±0.65     & 51.3±14.8                        & 106±8                                                      \\
{\textbf{~ RVI-LS}\\\textbf{~}}  &                                                &                                               &                                               &                                                   &                                                &                                                  &                                                 &                                                   &               &               &                                  &                                                            \\
\textbf{T1ADU}                   & 167±17                                         & 65±24                                         & 286±61                                        & {
  68.7±
  \\10.7
  }                            & 4.4±2.1                                        & 2.8±1.3                                          & 15.6±3.9                                        & 8.5±3.4                                           & 1.91±0.92     & 3.92±1.98     & 49.3±11.2                        & 410±32                                                     \\
\textbf{T1ADO}                   & 162±15                                         & 60±22                                         & 275±49                                        & {
  66.4±
  \\11.8
  }                            & 4.1±3.3                                        & 3.3±2.1                                          & 20.1±4.5                                        & 6.1±1.9                                           & 1.84±0.78     & 4.43±2.05     & 44.3±10.4                        & 452±38                                                     \\
\textbf{T1CHIL}                  & 159±13                                         & 62±18                                         & 264±45                                        & {
  62.3±
  \\12.3
  }                            & 7.9±4.4                                        & 2.2±1.5                                          & 22.3±5                                          & 5.3±1.4                                           & 2.08±0.98     & 4.34±2.12     & 37.3±9.8                         & 416±26                                                     \\
\textbf{T2ADU}                   & 150±28                                         & 63±20                                         & 279±55                                        & 70.1±11                                           & 5.2±3.5                                        & 2.8±1.1                                          & 14.6±4.2                                        & 7.3±2.2                                           & 1.94±0.81     & 3.85±1.87     & 55.4±22.4                        & 376±18                                                     
\end{tblr}}
\end{table*} 

\begin{table*}[t]
\centering
\caption{GLYCAEMIC CONTROL PERFORMANCE OF THE RESPECTIVE ALGORITHMS,\\ BY EMPLOYING THE PERSONALIZED, ROBUST INSULIN CONTROL POLICIES.}
\definecolor{Silver}{rgb}{0.749,0.749,0.749}
\definecolor{Black}{rgb}{0,0,0}
\scalebox{0.85}{\begin{tblr}{
  cells = {c},
  cell{1}{1} = {c=12}{},
  cell{7}{1} = {c=12}{},
  hlines,
  vlines = {Silver},
  vline{2} = {-}{Black},
}
{\textbf{~ R$\lambda$PI-LS}\\\textbf{~}} &                                                &                                               &                                               &                                                   &                                                &                                                  &                                                  &                                                   &               &               &                                  \\
\textbf{~}                       & {\textbf{BG}\\\textbf{mean}\\\textbf{[mg/dL]}} & {\textbf{BG}\\\textbf{min}\\\textbf{[mg/dL]}} & {\textbf{BG}\\\textbf{max}\\\textbf{[mg/dL]}} & {\textbf{\% in}\\\textbf{target}\\\textbf{range}} & {\textbf{\% in}\\\textbf{mild}\\\textbf{hypo}} & {\textbf{\% in}\\\textbf{severe}\\\textbf{hypo}} & {\textbf{\% in}\\\textbf{mild }\\\textbf{hyper}} & {\textbf{\% in}\\\textbf{severe}\\\textbf{hyper}} & \textbf{LBGI} & \textbf{HBGI} & {\textbf{TDI}\\\textbf{[U/day]}} \\
{\textbf{T1ADU}\\\textbf{~}}     & 139±12                                         & 90±15                                         & 200±17                                        & 88.8±4.6                                          & 0.4±0.1                                        & 0±0                                              & 10.8±4.5                                         & 0±0                                               & 0.25±0.12     & 1.41±0.45     & 49.2±12.4                        \\
\textbf{T1ADO}                   & 144±13                                         & 86±13                                         & 195±14                                        & 86.6±4.4                                          & 0.6±0.3                                        & 0±0                                              & 12.8±4.1                                         & 0±0                                               & 0.29±0.19     & 1.55±0.41     & 45.2±11.5                        \\
\textbf{T1CHIL}                  & 145±15                                         & 88±14                                         & 193±11                                        & 85.9±5.3                                          & 0.3±0.2                                        & 0±0                                              & 13.8±5.1                                         & 0±0                                               & 0.18±0.14     & 1.62±0.54     & 39.3±12.1                        \\
\textbf{T2ADU}                   & 137±15                                         & 85±12                                         & 204±18                                        & 87.1±6                                            & 0.7±0.5                                        & 0±0                                              & 12.2±5.5                                         & 0±0                                               & 0.36±0.28     & 1.49±0.76     & 56.3±15.8                        \\
{\textbf{~ RVI-LS}\\~
  }        &                                                &                                               &                                               &                                                   &                                                &                                                  &                                                  &                                                   &               &               &                                  \\
\textbf{T1ADU}                   & 149±14                                         & 75±23                                         & 255±54                                        & 78.6±6.8                                          & 2.3±1.8                                        & 0.8±0.3                                          & 10.9±3.1                                         & 7.4±1.6                                           & 1.31±0.6      & 2.85±1.14     & 53.3±15.8                        \\
\textbf{T1ADO}                   & 154±11                                         & 70±21                                         & 264±41                                        & 77±7.7                                            & 3.5±1.6                                        & 1.4±0.6                                          & 12.2±4.4                                         & 5.9±1.1                                           & 1.45±0.85     & 2.94±0.85     & 49.2±13.4                        \\
\textbf{T1CHIL}                  & 158±16                                         & 73±19                                         & 249±42                                        & 75.9±7.4                                          & 2.9±1.3                                        & 1.8±0.9                                          & 14.2±3.8                                         & 5.2±1.4                                           & 1.38±0.75     & 3.37±0.91     & 44.3±10.4                        \\
\textbf{T2ADU}                   & 152±39                                         & 71±25                                         & 261±50                                        & 79.4±7.6                                          & 4.1±2                                          & 1.2±0.5                                          & 9.1±3.2                                          & 6.2±1.9                                           & 1.61±0.93     & 2.81±1.05     & 59.9±25.5                        
\end{tblr}}
\end{table*} 
\begin{figure*}[h!]
\centering
\begin{multicols}{2}
  \hspace*{-1.2cm}\includegraphics[width=1.1\linewidth]{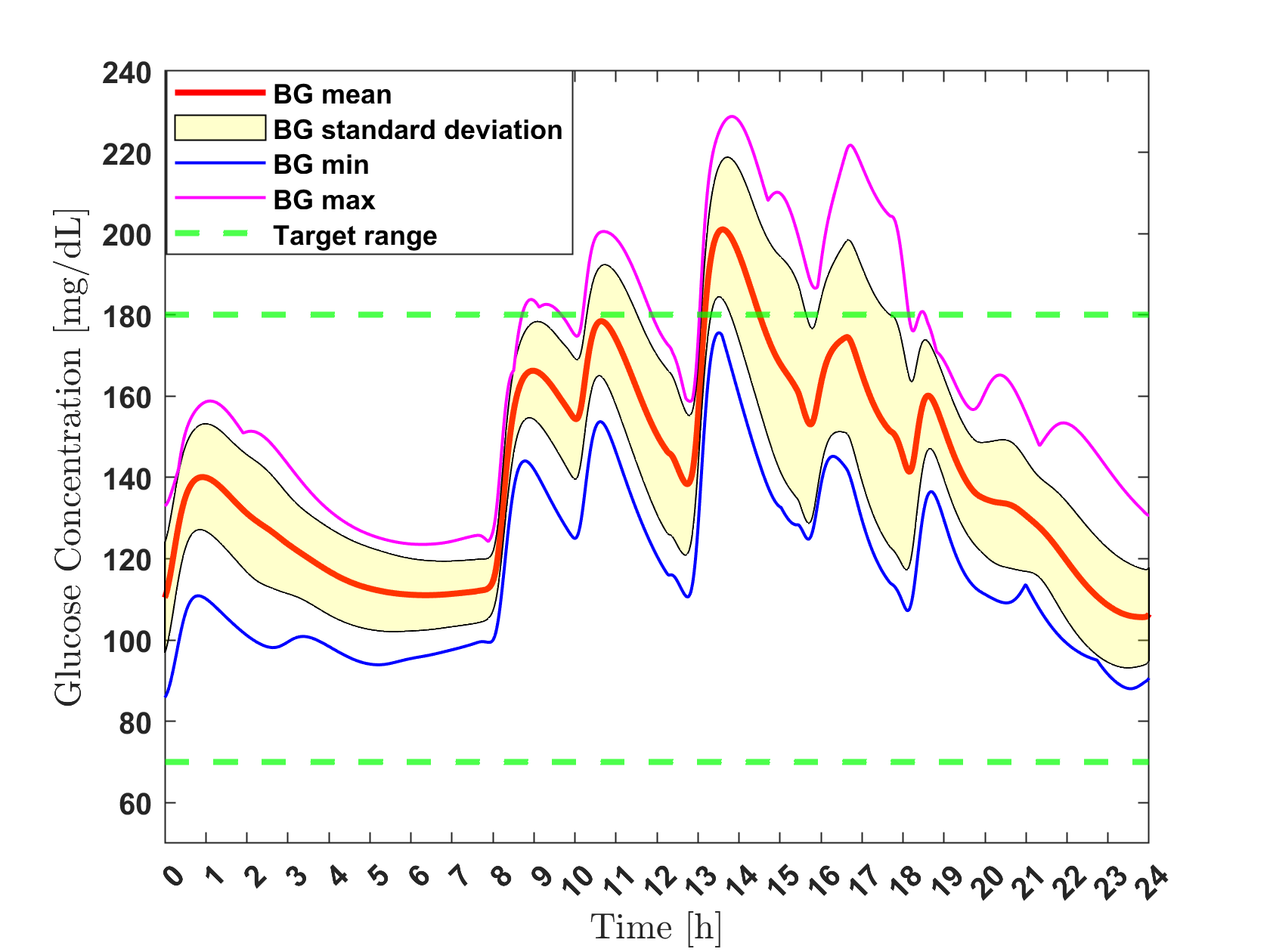}\par 
 \hspace*{-0.4cm} \includegraphics[width=1.1\linewidth]{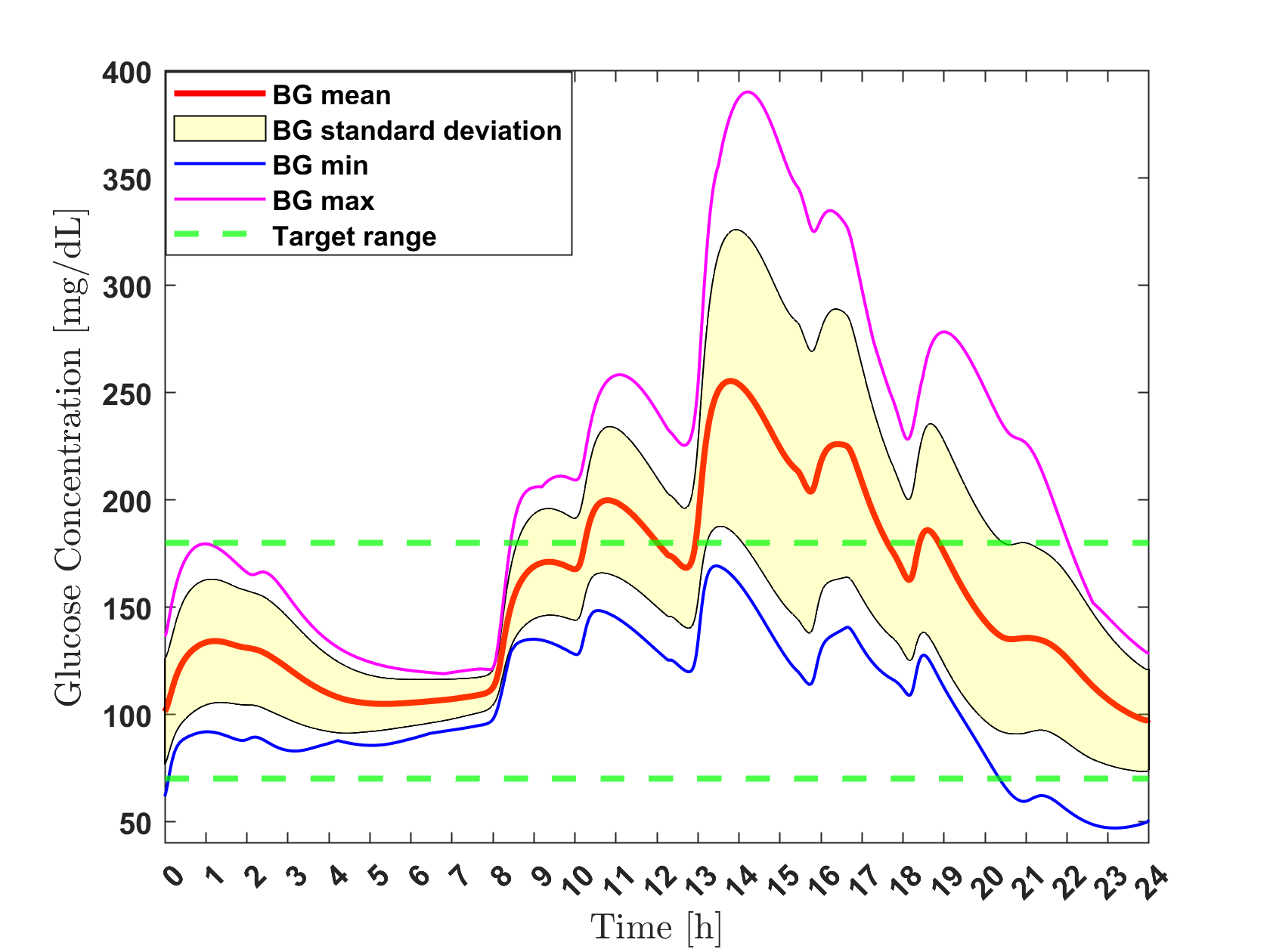}\par 
    \end{multicols}
\caption{Glycaemic behavior achieved by R$\lambda$PI-LS (left) and RVI-LS (right) algorithms, by employing the respective robust insulin policies to the T1DM adult population, during a day of a conducted in silico trial.}
\end{figure*}
Table I presents the results of the conducted in silico clinical studies for the virtual populations of T1DM adults (T1ADU), adolescents (T1ADO), children (T1CHIL) and T2DM adults (T2ADU). These results refer to the Learning Phase of the respective algorithms (i.e., until algorithmic convergence), for a value $\rho$ that satisfies the robustness conditions (Step 8 of Algorithm 2), that is $\rho_{\lambda PI}=7$ and $\rho_{VI}=11$. It is evident that R$\lambda$PI-LS converges significantly faster compared to RVI-LS (mean range of $106-134$ iterations compared to $376-452$ iterations for RVI-LS), which is also reflected on the required number of days for convergence (mean range of $53-67$ days compared to $188-226$ days for RVI-LS). Furthermore, R$\lambda$PI-LS enables superior glycaemic behavior for the entire duration of the Learning Phase, since it achieves significantly higher percentages of time spent in the target normoglycaemia (mean range of $77.3\%-82.6\%$ compared to $62.3\%-70.1\%$ for RVI-LS), significantly lower percentages of time spent in mild hypoglycaemia (mean range of $0.5\%-1\%$ compared to $4.1\%-7.9\%$ for RVI-LS), with no time spent in severe hypoglycaemia (compared to the mean range of $2.2\%-3.3\%$ for RVI-LS) and severe hyperglycaemia (compared to the mean range of $5.3\%-8.5\%$ for RVI-LS). The outstanding performance of R$\lambda$PI-LS is additionally shown based on the significantly lower values of LBGI (mean range of $0.21-0.32$ compared to $1.84-2.08$ for RVI-LS) and HBGI (mean range of $2.09-2.38$ compared to $3.85-4.43$ for RVI-LS). Similar significant improvements are observed in the related blood glucose metrics (BG mean, BG min and BG max).\\
After termination of the both algorithms, we conduct 2,000 in silico clinical trials for all considered T1DM and T2DM populations with a duration of 60 days, by now employing the converged robust insulin policy associated with each virtual subject. To fully evaluate the robustness of the converged control policies, we now consider higher dynamic variability on every meal and exercise occurence of the nominal configuration, with a range defined as follows: $1)$ $[-60,60]$ minutes on the meal starting time, $2)$ $[-50\%,50\%]$ on the amount of CHO, $3)$ $[-50\%,50\%]$ on meal duration, $4)$ $[-60,60]$ minutes on exercise time, $5)$ random selection of the exercise intensity mode from [light, moderate, intense] and $6)$ $[-50\%,50\%]$ on exercise duration. The considered variability profile follows uniform distributions. Table II presents the results. As expected, the conveged robust insulin policies computed by R$\lambda$PI-LS and RVI-LS achieve significantly improved glycaemic control compared to the behavior observed during the Learning Phase. It is clear that R$\lambda$PI-LS accompishes critically better glycaemic control compared to RVI-LS, with significantly higher time spent in the target range (mean range of $85.9\%-88.8\%$ compared to $75.9\%-79.4\%$ for RVI-LS) and significantly less time spent in mild hypoglyceamia (mean range of $0.3\%-0.7\%$ compared to $2.3\%-4.1\%$ for RVI-LS). Furthermore, the robust insulin control policies computed by $R\lambda$PI-LS result in no time spent in severe hypoglycaemia (compared to the mean range of $0.8\%-1.8\%$ for RVI-LS) and severe hyperglycaemia (compared to the mean range of $5.2\%-7.4\%$ for RVI-LS). In a similar fashion, crucial improvements are also observed in the values of LBGI, HBGI and blood glucose metrics (BG mean, BG min and BG max) for both T1DM and T2DM populations.\\
To enable a better understanding of the glycaemic control achieved by the two algorithms, Figure 1 shows the resulting glycaemic behavior achieved by employing the converged robust insulin policies to the T1DM adult population, during a day of a conducted in silico trial. The meal and exercise configuration for that specific day is defined as follows: The meals start at $[07\text{:}30, 09\text{:}45, 12\text{:}15, 15\text{:}25, 17\text{:}40, 23\text{:}45]$, contain $[60, 40, 100, 30, 65, 30]$ grams of carbohydrates (CHO) and last $[40, 25, 45, 25, 35, 15]$ minutes respectively. Furthermore, the exercise session starts at $16\text{:}35$, is characterized by moderate intensity and lasts 30 minutes. Overall, it is evident that R$\lambda$PI-LS achieves superior glycaemic control for both T1DM and T2DM populations, under the existence of completely unannounced, highly uncertain meal and exercise scenarios. 
\section{Conclusion}
In this work, we presented a novel robust reinforcement learning control algorithm, with critical applications to fully-automated drug delivery for personalized medicine. We proposed a novel robust tracking control scheme which can ensure a sufficiently small tracking control error for a discrete-time uncertain nonlinear system. This can be achieved by solving a suitably defined, parameterized optimal tracking control problem. The associated optimal tracking control policy can be computed based on a novel Q-function-based variant of the $\lambda$-PI algorithm.  The proposed algorithm is shown to benefit from rigorous theoretical guarantees and avoid common drawbacks of PI and VI. By utilizing a critic-only LS implementation approach, we evaluate the performance and reliability of the overall data-driven robust learning control algorithm to the challenging problem of fully-automated, insulin-based, closed-loop glucose control for patients diagnosed with T1DM and T2DM. To accomplish this, a U.S. FDA-accepted metabolic simulator is employed to conduct an extensive in silico campaign on completely unannounced meal and exercise settings.\\
As a future work, we aim to theoretically extend and assess the proposed algorithmic framework to the case of uncertain multiplayer nonlinear games, by merging robust reinforcement learning and game theoretical control methods. Furthermore, we plan to proceed with a microcontroller implementation of the derived robust algorithm, with the final goal of producing a prototype that can be evaluated in future clinical trials on real patients.
\appendices
\section{Proof of Theorem 1}
To simplify presentation, we define the following compact notation
\begin{align*}
G^{\mu^{\star}}(x,r) =& G\big(x,r,\mu^{\star}(x,r)\big)\nonumber\\
G^{\star,\mu^{\star}}(x,r) =& G^{\star}\big(x,r,\mu^{\star}(x,r)\big)\nonumber\\
\nabla_x G^{\star,\mu^{\star}}(x,r) =& \nabla_x G^{\star}\big(x,r,\mu^{\star}(x,r)\big)\nonumber\\
\nabla_{xx} G^{\star,\mu^{\star}}(x,r) =& \nabla_{xx} G^{\star}\big(x,r,\mu^{\star}(x,r)\big)
\end{align*}
for generic $G:\mathcal{Z}\rightarrow \mathbb{R}_{+}$, $G^{\star}:\mathcal{Z}\rightarrow \mathbb{R}_{+}$, $\nabla_x G^{\star}:\mathcal{Z}\rightarrow \mathbb{R}^{n}$ and $\nabla_{xx} G^{\star}:\mathcal{Z}\rightarrow \mathbb{R}^{n\times n}$.
\begin{itemize}
\item[1)] We employ $\gamma^{s}Q^{\star,\mu^{\star}}(x_s,r_s)$ as a candidate Lyapunov function. The difference equation can then be written as $D\big(\gamma^{s}Q^{\star,\mu^{\star}}(x_s,r_s)\big)=\gamma^{s+1}Q^{\star,\mu^{\star}}(x_{s+1},r_{s+1})-\gamma^{s}Q^{\star,\mu^{\star}}(x_{s},r_{s})$. Hence,
\begin{align}
\label{eq:21}
D\big(\gamma^{s}Q^{\star,\mu^{\star}}(x_s,r_s)\big)=& -\gamma^{s}\big[l^{\mu^{\star}}(x_s,r_s)+\Gamma^{\mu^{\star}}(x_s,r_s)\big].
\end{align}
It is clear that $D\big(\gamma^{s}Q^{\star,\mu^{\star}}(x_s,r_s)\big) \leq 0$. For the case of $\gamma=1$ (can only be applied if $\lim_{s\to\infty} r_{s}\rightarrow 0$), based on Barbalat's Extension Lemma \cite{b71}, \cite[p. 113]{b87} the states $x$ of  the nominal system \eqref{eq:3} converge in a region where it holds that $\lim_{s\to \infty} D\big(\gamma^{s}Q^{\star,\mu^{\star}}(x_s,r_s)\big) \rightarrow 0$. According to \eqref{eq:21}, local asymptotic stability of the tracking error $e_s$ on $\mathcal{X}$ is therefore achieved for the closed-loop nominal system \eqref{eq:3} under $\mu^{\star}(x,r)$, meaning that $\lim_{s\to\infty} x_{s}\rightarrow 0$. Otherwise, following \cite{b88}, the tracking error $e_s$ can be made sufficiently small for \eqref{eq:3} by utilizing a value of $\gamma$ sufficiently close to 1.
\item[2)] Similar to the proof for part 1), we use $\gamma^{s}Q^{\star,\mu^{\star}}(x_s,r_s)$ as a candidate Lyapunov function. Then, the difference equation can be written as $\tilde{D}\big(\gamma^{s}Q^{\star,\mu^{\star}}(x_s,r_s)\big)=\gamma^{s+1}Q^{\star,\mu^{\star}}(\tilde{x}_{s+1},r_{s+1})- \gamma^{s}Q^{\star,\mu^{\star}}(x_{s},r_{s})$ where $\tilde{x}_{s+1}=f\big(x_s,\mu^{\star}(x_s,r_s)\big)+d(x_s)$. By taking the second-order Taylor approximation of $Q^{\star,\mu^{\star}}(\tilde{x}_{s+1},r_{s+1})$ at the operating point $x_{s+1}=f\big(x_s,\mu^{\star}(x_s,r_s)\big)$ related to the nominal system \eqref{eq:3}, the difference equation yields 
\begin{align}
\label{eq:22}
&\tilde{D}\big(\gamma^{s}Q^{\star,\mu^{\star}}(x_s,r_s)\big) =\nonumber \\ &\gamma^{s}\bigg[-(x_s-r_s)^{T}S(x_s-r_s) -[\mu^{\star}(x_s,r_s)]^{T}R\mu^{\star}(x_s,r_s)\nonumber \\
&-\rho^{2} \Delta^{2}(x_s)-\frac{1}{4}\gamma \big \| \nabla_x Q^{\star,\mu^{\star}}(x_{s+1},r_{s+1})\big \|^{2}_{2}\nonumber \\
&+ \gamma \big[\nabla_x Q^{\star,\mu^{\star}}(x_{s+1},r_{s+1})\big]^{T}d(x_s) \nonumber \\
&+\frac{1}{2}\gamma [d(x_s)]^{T}\nabla_{xx} Q^{\star,\mu^{\star}}(x_{s+1},r_{s+1})d(x_s)\bigg]
\end{align}
for which it is assumed that the third and high-order terms are small and can be neglected, similar to \cite{b10,b14,b15}. By adding and subtracting the term $\gamma^{s+1}\big\|d(x_s)\big\|^{2}_{2}$ in the right hand side of \eqref{eq:22}, we get that
\begin{align}
\label{eq:23}
&\tilde{D}\big(\gamma^{s}Q^{\star,\mu^{\star}}(x_s,r_s)\big)=\nonumber \\
& \gamma^{s}\bigg[-(x_s-r_s)^{T}S(x_s-r_s) -[\mu^{\star}(x_s,r_s)]^{T}R\mu^{\star}(x_s,r_s)\nonumber \\
&-\gamma \big[\frac{1}{2}\nabla_x Q^{\star,\mu^{\star}}(x_{s+1},r_{s+1})-d(x_s)\big]^{T}\nonumber \\
& \cdot \big[\frac{1}{2}\nabla_x Q^{\star,\mu^{\star}}(x_{s+1},r_{s+1})-d(x_s)\big]\nonumber \\
& -\big[\rho^{2} \Delta^{2}(x_s)- \gamma \big\|d(x_s)\big\|^{2}_{2}\nonumber \\
&-\frac{1}{2}\gamma [d(x_s)]^{T}\nabla_{xx} Q^{\star,\mu^{\star}}(x_{s+1},r_{s+1})d(x_s)\big]\bigg].
\end{align}
If condition \eqref{eq:14} holds, then $\tilde{D}\big(\gamma^{s}Q^{\star,\mu^{\star}}(x_s,r_s)\big)\leq 0$. Similar to the proof for part 1), for the case $\gamma=1$ (can only be utilized if $\lim_{s\to\infty} r_{s}\rightarrow 0$), based on Barbalat's Extension Lemma \cite{b71}, \cite[p. 113]{b87}, the states $x$ of the uncertain system \eqref{eq:1} converge in a region where it holds that $\lim_{s\to \infty} \tilde{D}\big(\gamma^{s}Q^{\star,\mu^{\star}}(x_s,r_s)\big) \rightarrow 0$. Based on \eqref{eq:23}, local asymptotic stability of the tracking error $e_s$ on $\mathcal{X}$ is therefore achieved for the closed-loop uncertain system \eqref{eq:1} under $\mu^{\star}(x,r)$, i.e., $\lim_{s\to\infty} x_{s}\rightarrow 0$. Otherwise, by following the theoretical results in \cite{b88}, the tracking error $e_s$ can be made sufficiently small for \eqref{eq:1} by using a value of $\gamma$ sufficiently close to 1.
\end{itemize}

\section{Proof of Theorem 2}
Before proceeding with the proof of Theorem 2, we first prove the following important result.
\begin{lemma}
Let Assumptions 1 to 3 hold and $Q(x,r,a)\geq 0$ satisfy
\begin{equation}
\label{eq:24}
Q(x,r,a) \geq  l(x,r,a) + \Gamma^{Q}(x,r,a) + \gamma Q\big(x',r',\mu(x',r')\big),
\end{equation}
where $\Gamma^{Q}(x,r,a) =\rho^{2}\Delta^{2}(x)$$ + \frac{1}{4}\gamma \big\|\nabla_x Q\big(x',r',\mu(x',r')\big)\big\|^{2}_{2}$ and $\mu(x,r) =\underset{u}{\mathrm{argmin}}\enspace Q(x,r,u)$. Furthermore, consider the sequence $\{\Xi^{j}(x,r,a)\}_{j\in \mathbb{N}}$ generated by 
\begin{align}
\label{eq:25}
&\Xi^{j+1}(x,r,a) = l(x,r,a) + \Gamma^{Q}(x,r,a) \nonumber \\
&+\lambda \gamma \Xi^{j}\big(x',r',\mu(x',r')\big)+ (1-\lambda)\gamma Q\big(x',r',\mu(x',r')\big)
\end{align}
with $\Xi^{0}(x,r,a) = Q(x,r,a)$, Moreover, let $\underline{Q}(\cdot,\cdot,\cdot)$ be the solution of
\begin{align}
\label{eq:26}
&\underline{Q}(x,r,a) = l(x,r,a) + \Gamma^{Q}(x,r,a) \nonumber \\
&+\lambda \gamma \underline{Q}\big(x',r',\mu(x',r')\big)+ (1-\lambda)\gamma Q\big(x',r',\mu(x',r')\big).
\end{align}
Then, for all $(x,r,a)\in \mathcal{Z}$
\begin{itemize}
\item[1)] $\Xi^{j+1}(x,r,a) \leq \Xi^{j}(x,r,a)$.
\item [2)] $\lim_{j \to \infty} \Xi^{j}(x,r,a) = \underline{Q}(x,r,a)$.
\end{itemize}
\end{lemma}
\begin{proof}
To simplify presentation, we define the following compact notation
\begin{align}
\label{eq:27}
G^{\mu}(x,r) =& G\big(x,r,\mu(x,r)\big)\nonumber\\
G^{j,\mu}(x,r) =& G^{j}\big(x,r,\mu(x,r)\big)
\end{align}
for generic $G:\mathcal{Z}\rightarrow \mathbb{R}_{+}$, $G^{j}:\mathcal{Z}\rightarrow \mathbb{R}_{+}$ and $j\in \mathbb{N}_0$.
\begin{itemize}
\item[1)] We prove the desired statement based on mathematical induction. For $j=0$,
\begin{align*}
\Xi^{1}(x,r,a) =& l(x,r,a) + \Gamma^{Q}(x,r,a)+\lambda \gamma \Xi^{0,\mu}(x',r') \\
&+(1-\lambda)\gamma Q^{\mu}(x',r') \\
=& l(x,r,a) + \Gamma^{Q}(x,r,a)+\gamma Q^{\mu}(x',r') \\
\overset{\eqref{eq:24}}{\leq} & Q(x,r,a) \\
=& \Xi^{0}(x,r,a).
\end{align*}
Assume that $\Xi^{j}(x,r,a)\leq \Xi^{j-1}(x,r,a)$. Based on \eqref{eq:25},
\begin{align*}
\Xi^{j+1}(x,r,a) =&  l(x,r,a) + \Gamma^{Q}(x,r,a)+\lambda \gamma \Xi^{j,\mu}(x',r') \\
&+(1-\lambda)\gamma Q^{\mu}(x',r') \\
\leq &  l(x,r,a) + \Gamma^{Q}(x,r,a)+\lambda \gamma \Xi^{j-1,\mu}(x',r') \\
&+(1-\lambda)\gamma Q^{\mu}(x',r') \\
=& \Xi^{j}(x,r,a).
\end{align*}
\item[2)] According to \eqref{eq:25} and \eqref{eq:26},
\begin{align*}
&\Xi^{j+1}(x,r,a) - \underline{Q}(x,r,a)\nonumber \\
&=\lambda \gamma \big[\Xi^{j,\mu}(x',r')-\underline{Q}^{\mu}(x',r')\big]\nonumber \\
&= \lambda \gamma \big[\Xi^{j+1,\mu}(x',r')-\underline{Q}^{\mu}(x',r')\big]\nonumber \\
&+ \lambda \gamma \big[\Xi^{j,\mu}(x',r')-\Xi^{j+1,\mu}(x',r')\big].
\end{align*}
Then, for all $(x,r,a)\in \mathcal{Z}$,
\begin{align*}
&\big\|\Xi^{j+1}(x,r,a) - \underline{Q}(x,r,a)\big\|_{\infty}\nonumber \\
& \leq  \lambda \gamma \big\|\Xi^{j+1}(x,r,a)-\underline{Q}(x,r,a)\big\|_{\infty} \\
&+ \lambda \gamma \big\|\Xi^{j}(x,r,a)-\Xi^{j+1}(x,r,a)\big\|_{\infty},
\end{align*}
i.e.,
\begin{align*}
&\big\|\Xi^{j+1}(x,r,a) - \underline{Q}(x,r,a)\big\|_{\infty}\\
& \leq \frac{\lambda \gamma}{(1-\lambda \gamma)}\big\|\Xi^{j}(x,r,a)-\Xi^{j+1}(x,r,a)\big\|_{\infty}.
\end{align*}
Hence, we get that
\begin{align*}
&\big\|\Xi^{j}(x,r,a) - \underline{Q}(x,r,a)\big\|_{\infty}\\
\leq & \frac{\lambda \gamma}{(1-\lambda \gamma)}\big\|\Xi^{j-1}(x,r,a)-\Xi^{j}(x,r,a)\big\|_{\infty} \\
=& \frac{(\lambda \gamma)^{2}}{(1-\lambda \gamma)}\big\|\Xi^{j-2}(x,r,a)-\Xi^{j-1}(x,r,a)\big\|_{\infty} \\
=& \ldots \\
=&  \frac{(\lambda \gamma)^{j}}{(1-\lambda \gamma)}\big\|\Xi^{0}(x,r,a)-\Xi^{1}(x,r,a)\big\|_{\infty}.
\end{align*}
Therefore, $\lim_{j \to \infty} \big\|\Xi^{j}(x,r,a) - \underline{Q}(x,r,a)\big\|_{\infty} =0$, i.e., $\lim_{j\to \infty} \Xi^{j}(x,r,a) = \underline{Q}(x,r,a)$.
\end{itemize}
\end{proof}
We are now ready to proceed with the proof of Theorem 2. To simplify presentation, we define the compact notation
\begin{align*}
G^{\mu^{i}}(x,r) =& G\big(x,r,\mu^{i}(x,r)\big)\\
G^{i,\mu^{i}}(x,r) =& G^{i}\big(x,r,\mu^{i}(x,r)\big)
\end{align*}
for generic $G:\mathcal{Z}\rightarrow \mathbb{R}_{+}$, $G^{i}:\mathcal{Z}\rightarrow \mathbb{R}_{+}$ and $i\in \mathbb{N}_0$.
\begin{itemize}
\item[1)] Mathematical induction is employed to prove \eqref{eq:15}. For the case $i=0$, let $\Xi^{0}(x,r,a)= Q(x,r,a)=Q^{0}(x,r,a)$ and $\mu(x,r) = \mu^{0}(x,r)$. According to Lemma 1, the sequence $\{\Xi^{j}(x,r,a)\}_{j\in \mathbb{N}_0}$ is non-increasing and $\lim_{j\to \infty} \Xi^{j}(x,r,a) = Q^{1}(x,r,a)$. Then,
\begin{align*}
Q^{1}(x,r,a) =& \lim_{j\to \infty} \Xi^{j}(x,r,a) \nonumber \\
\leq & \Xi^{1}(x,r,a) \nonumber \\
\overset{\eqref{eq:25}}{=}& l(x,r,a) + \Gamma^{0}(x,r,a) +\lambda \gamma Q^{0,\mu^{0}}(x',r')\nonumber \\
&+ (1-\lambda)\gamma Q^{0,\mu^{0}}(x',r')\nonumber \\
= & l(x,r,a) + \Gamma^{0}(x,r,a) + \gamma Q^{0,\mu^{0}}(x',r')\nonumber \\
\overset{\eqref{eq:13}}{\leq} & Q^{0}(x,r,a).
\end{align*}
This means that \eqref{eq:15} holds for $i=0$. We then assume that \eqref{eq:15} holds for $i-1$
\begin{align}
\label{eq:28}
Q^{i}(x,r,a) \leq & l(x,r,a) +\Gamma^{i-1}(x,r,a) + \gamma Q^{i-1,\mu^{i-1}}(x',r') \nonumber \\
\leq & Q^{i-1}(x,r,a).
\end{align}
According to \eqref{eq:12},
\begin{align}
\label{eq:29}
Q^{i}(x,r,a) =& l(x,r,a) + \Gamma^{i-1}(x,r,a) + \lambda \gamma Q^{i,\mu^{i-1}}(x',r')\nonumber \\
&+(1-\lambda)\gamma Q^{i-1,\mu^{i-1}}(x',r') \nonumber \\
\overset{\eqref{eq:28}}{\geq} &  l(x,r,a) + \Gamma^{i-1}(x,r,a) + \lambda \gamma Q^{i,\mu^{i-1}}(x',r')\nonumber \\
&+(1-\lambda)\gamma Q^{i,\mu^{i-1}}(x',r') \nonumber \\
=& l(x,r,a) + \Gamma^{i-1}(x,r,a) + \gamma Q^{i,\mu^{i-1}}(x',r') \nonumber \\
\overset{\eqref{eq:11}}{\geq} &  l(x,r,a) + \Gamma^{i-1}(x,r,a) + \gamma Q^{i,\mu^{i}}(x',r').
\end{align}
If condition \eqref{eq:14} holds, then \eqref{eq:29} becomes 
\begin{align}
\label{eq:30}
Q^{i}(x,r,a) \geq & l(x,r,a) + \Gamma^{i}(x,r,a) + \gamma Q^{i,\mu^{i}}(x',r').
\end{align}
Next, let $\Xi^{0}(x,r,a)= Q(x,r,a)=Q^{i}(x,r,a)$ and $\mu(x,r) = \mu^{i}(x,r)$. According to Lemma 1, the sequence $\{\Xi^{j}(x,r,a)\}_{j\in \mathbb{N}_0}$ is non-increasing and $\lim_{j\to \infty} \Xi^{j}(x,r,a) = Q^{i+1}(x,r,a)$. Therefore,
\begin{align}
\label{eq:31}
Q^{i+1}(x,r,a) =& \lim_{j\to \infty} \Xi^{j}(x,r,a) \nonumber \\
\leq & \Xi^{1}(x,r,a) \nonumber \\
\overset{\eqref{eq:25}}{=}& l(x,r,a) + \Gamma^{i}(x,r,a) +\lambda \gamma Q^{i,\mu^{i}}(x',r')\nonumber \\
&+ (1-\lambda)\gamma Q^{i,\mu^{i}}(x',r')\nonumber \\
=& l(x,r,a) + \Gamma^{i}(x,r,a) + \gamma Q^{i,\mu^{i}}(x',r').
\end{align}
Combining \eqref{eq:30} and \eqref{eq:31} proves \eqref{eq:15} for all $i\in \mathbb{N}_0$ and $(x,r,a)\in \mathcal{Z}$.
\item[2)] Based on \eqref{eq:15}, the sequence $\{Q^{i}(x,r,a)\}_{i\in \mathbb{N}_0}$ is non-increasing and, due to the fact that $l$ and $\Gamma^{i}$ are non-negative,  lower bounded by $0$ for all $i$. Therefore, it has a point-wise limit $Q^{\infty}(x,r,a) = \lim_{i\rightarrow \infty} Q^{i}(x,r,a)$. By defining $\mu^{\infty}(x,r) = \underset{u}{\mathrm{argmin }}\enspace Q^{\infty, u}(x,r)$ and computing the limit of \eqref{eq:15}, we get
\begin{align*}
Q^{\infty}(x,r,a) \leq & l(x,r,a)+ \Gamma^{\infty}(x,r,a)+ \gamma Q^{\infty,\mu^{\infty}}(x',r')\nonumber \\
\leq & Q^{\infty}(x,r,a),\nonumber
\end{align*}
i.e., 
\begin{align}
\label{eq:32}
Q^{\infty}(x,r,a)=& l(x,r,a)+\Gamma^{\infty}(x,r,a)+ \gamma Q^{\infty,\mu^{\infty}}(x',r').
\end{align}
Based on the uniqueness of solutions to the Bellman equation \cite{b72,b73,b89,b90}, we finally get that \eqref{eq:32} is essentially \eqref{eq:6}, i.e., $Q^{\infty}(x,r,a)=Q^{\star}(x,r,a)$ and therefore $\mu^{\infty}(x,r)=\mu^{\star}(x,r)$.
\end{itemize}

\section{Proof of Corrolary 1}
We use the compact notation \eqref{eq:27} to present the proof. Consider the sequence $\{\Xi_{\Lambda}^{j}(x,r,a)\}_{j\in \mathbb{N}}$ generated by
\begin{align}
\label{eq:33}
\Xi_{\Lambda}^{j+1}(x,r,a) =& l(x,r,a) + \Gamma^{Q}(x,r,a) +\Lambda \gamma \Xi_{\Lambda}^{j,\mu}(x',r')\nonumber \\
&+(1-\Lambda)\gamma Q^{\mu}(x',r'),
\end{align}
with $\Xi_{\Lambda}^{0}(x,r,a)=Q(x,r,a)$ and $\Lambda \in \{\lambda_1,\lambda_2\}$. According to Lemma 1, the sequence $\{\Xi_{\Lambda}^{j}(x,r,a)\}_{j\in \mathbb{N}_0} $ is non-increasing. We will now prove that $\Xi_{\lambda_2}^{j}(x,r,a) \leq \Xi_{\lambda_1}^{j}(x,r,a)$ for all $j\in \mathbb{N}_0$ and $(x,r,a)\in \mathcal{Z}$ based on mathematical induction. We firstly note that $\Xi_{\lambda_1}^{0}(x,r,a)= \Xi_{\lambda_2}^{0}(x,r,a)=Q(x,r,a)$. We then assume that $\Xi_{\lambda_2}^{j-1}(x,r,a)\leq \Xi_{\lambda_1}^{j-1}(x,r,a)$. By realizing \eqref{eq:33} for $\lambda_1$, we get
\begin{align}
\label{eq:34}
\Xi_{\lambda_1}^{j}(x,r,a) =& l(x,r,a) + \Gamma^{Q}(x,r,a) +\lambda_1 \gamma \Xi_{\lambda_1}^{j-1,\mu}(x',r')\nonumber \\
&+(1-\lambda_1)\gamma Q^{\mu}(x',r')\nonumber \\
\geq &  l(x,r,a) + \Gamma^{Q}(x,r,a) +\lambda_1 \gamma \Xi_{\lambda_2}^{j-1,\mu}(x',r')\nonumber \\
&+(1-\lambda_1)\gamma Q^{\mu}(x',r').
\end{align}
Hence, by now realizing \eqref{eq:33} for $\lambda_1$ and $\lambda_2$ and applying \eqref{eq:34}, we have that
\begin{align}
\label{eq:35}
&\Xi_{\lambda_1}^{j}(x,r,a)-\Xi_{\lambda_2}^{j}(x,r,a) \geq \nonumber \\ &(\lambda_1-\lambda_2)\gamma \Xi_{\lambda_2}^{j-1,\mu}(x',r')- (\lambda_1-\lambda_2)\gamma Q^{\mu}(x',r')\nonumber \\
&= (\lambda_2-\lambda_1)\gamma \big[Q^{\mu}(x',r')-\Xi_{\lambda_2}^{j-1,\mu}(x',r')\big]\nonumber \\
&= (\lambda_2-\lambda_1)\gamma \big[\Xi_{\lambda_2}^{0,\mu}(x',r')-\Xi_{\lambda_2}^{j-1,\mu}(x',r')\big]\nonumber \\
&\geq 0,
\end{align}
which means that $\Xi_{\lambda_2}^{j}(x,r,a) \leq \Xi_{\lambda_1}^{j}(x,r,a)$ for all $j\in \mathbb{N}_0$ and $(x,r,a)\in \mathcal{Z}$. According to Lemma 1, we have that $Q_{\Lambda}(x,r,a) = \lim_{j\to \infty} \Xi_{\Lambda}^{j}(x,r,a)$ for $\Lambda \in \{\lambda_1,\lambda_2\}$. Therefore, it follows from \eqref{eq:35} that $Q_{\lambda_2}(x,r,a) \leq Q_{\lambda_1}(x,r,a)$ for all $(x,r,a)\in \mathcal{Z}$.

\begin{IEEEbiography}[{\includegraphics[width=1in,height=1.25in,clip,keepaspectratio]{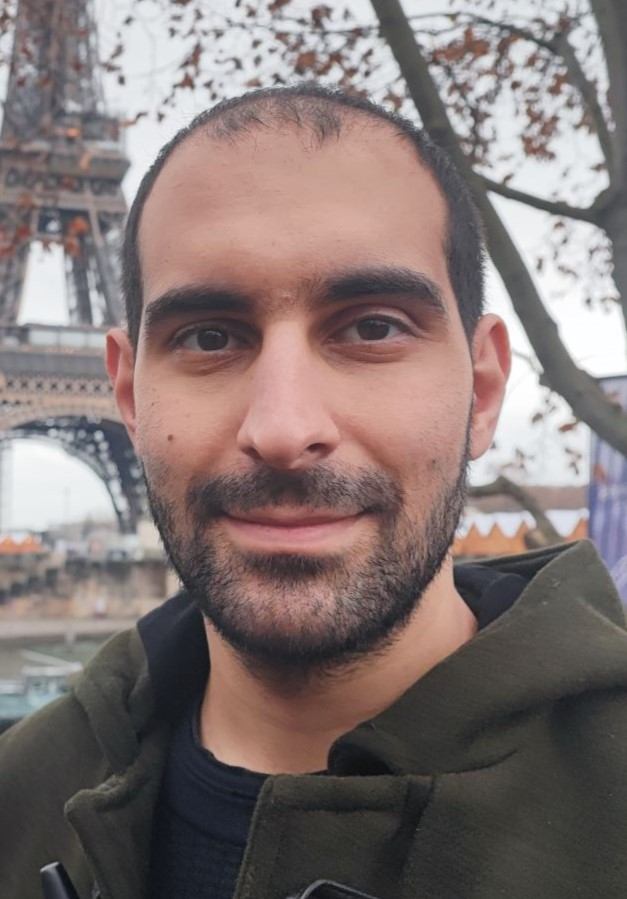}}]{Alexandros Tanzanakis}
received a 5-year Diploma (equivalent to a Master's degree) in Electrical and Computer Engineering with the highest honours from the Technical University of Crete, Greece in 2016, and a PhD in Information Technology and Electrical Engineering from ETH Zurich, Switzerland in 2023. His research interests include advanced topics in learning-based control, data-driven reinforcement learning, game theory and multiagent systems, as well as intelligent biomedical control with emphasis on the design of novel, fully-automated, personalized, closed-loop drug delivery systems.
\end{IEEEbiography}

\begin{IEEEbiography}[{\includegraphics[width=1in,height=1.25in,clip,keepaspectratio]{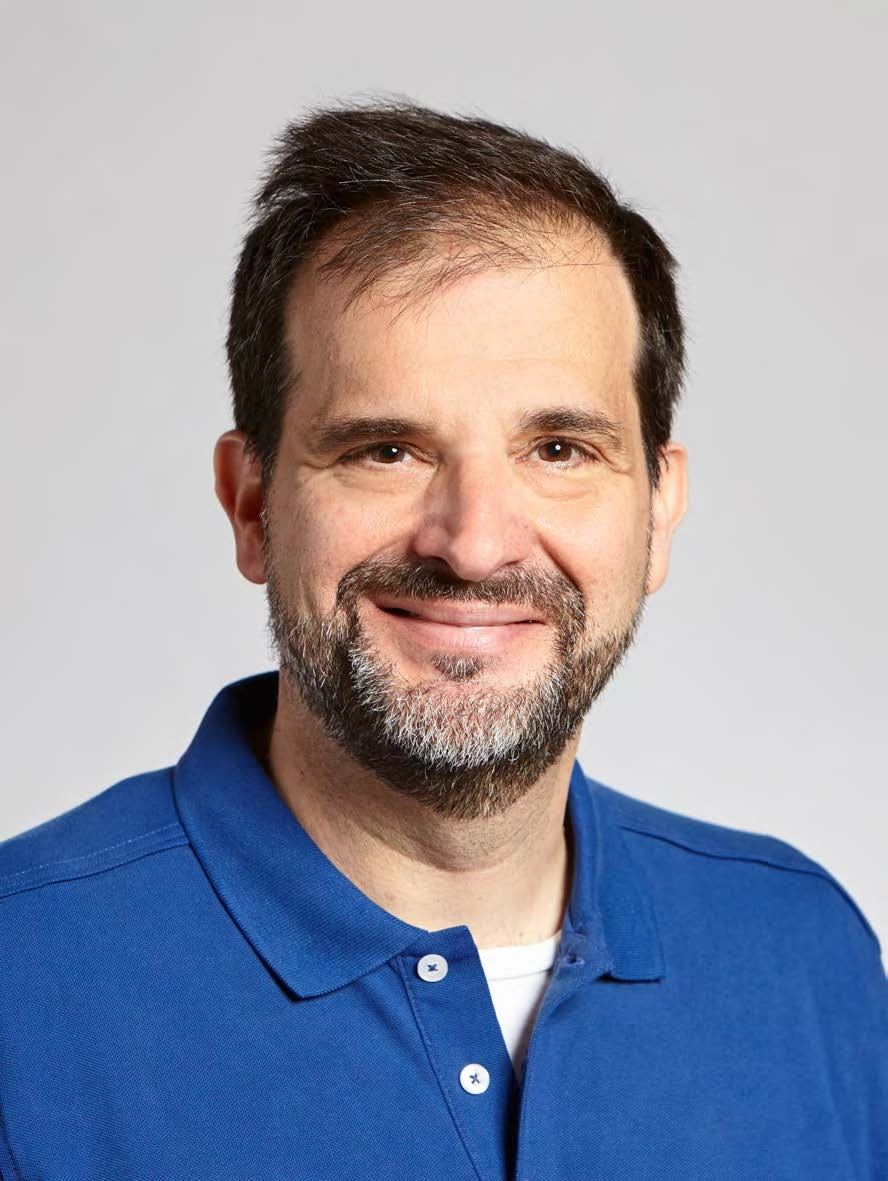}}]{John Lygeros}
received a B.Eng. degree in 1990 and an M.Sc. degree in 1991 from Imperial College, London, U.K. and a Ph.D. degree in 1996 at the University of California, Berkeley. After research appointments at M.I.T., U.C. Berkeley and SRI International, he joined the University of Cambridge in 2000 as a University Lecturer. Between March 2003 and July 2006 he was an Assistant Professor at the Department of Electrical and Computer Engineering, University of Patras, Greece. In July 2006 he joined the Automatic Control Laboratory at ETH Zurich where he is currently serving as the Professor for Computation and Control and the Head of the laboratory. His research interests include modelling, analysis, and control of large-scale systems, with applications to biochemical networks, energy systems, transportation, and industrial processes. John Lygeros is a Fellow of IEEE, and a member of IET and the Technical Chamber of Greece. Since 2013 he is serving as the Vice-President Finances and a Council Member of the International Federation of Automatic Control and since 2020 as the Director of the National Center of Competence in Research "Dependable Ubiquitous Automation" (NCCR Automation).
\end{IEEEbiography}

\end{document}